\documentclass[runningheads,a4paper]{llncs}

\usepackage{ltl}
\usepackage{wrapfig}
\usepackage{booktabs}
\usepackage{tikz}
\usetikzlibrary{calc}
\usepackage{color}
\usepackage{scalefnt}
\usepackage{makecell}
\usepackage{array}

\usepackage{amssymb}
\usepackage{graphicx}
\usepackage[noend]{algorithmic}
\usepackage{algorithm2e}
\usepackage{stmaryrd}

\usepackage{url}
\urldef{\mailsa}\path|{alfred.hofmann, ursula.barth, ingrid.haas, frank.holzwarth,|
\urldef{\mailsb}\path|anna.kramer, leonie.kunz, christine.reiss, nicole.sator,|
\urldef{\mailsc}\path|erika.siebert-cole, peter.strasser, lncs}@springer.com|

\newcommand{\todo}[1]{}

\newcommand{\AP}{\textit{AP} }

\newcommand{\A}[0]{\mathcal A}

\newcommand{\nats}{\mathbb{N}}

\newcommand{\underapprox}[0]{\subseteq_n}
\newcommand{\overapprox}[0]{\supseteq_n}

\begin{document}

\mainmatter  

\title{Approximate Automata\\for Omega-regular Languages
\thanks{This work was partially supported by the German Research Foundation (DFG) as part of the Collaborative Research Center “Methods and Tools for Understanding and Controlling Privacy” (CRC 1223) and the Collaborative Research Center “Foundations of Perspicuous Software Systems” (TRR 248, 389792660), and by the European Research Council (ERC) Grant OSARES (No. 683300).}}

\titlerunning{}

%
%
\author{Rayna Dimitrova\inst{1} \and Bernd Finkbeiner\inst{2} \and Hazem Torfah\inst{2}}
\authorrunning{}

\institute{University of Leicester \and Saarland University}

%
%

\toctitle{Lecture Notes in Computer Science}
\tocauthor{Authors' Instructions}
\maketitle

\begin{abstract}
Automata over infinite words, also known as $\omega$-automata, play a key role in the verification and synthesis of reactive systems. The spectrum of $\omega$-automata is defined by two characteristics: the \emph{acceptance condition} (e.g. B\"uchi or parity) and the \emph{determinism} (e.g., deterministic or nondeterministic) of an automaton. These characteristics play a crucial role in applications of automata theory. For example, certain acceptance conditions can be handled more efficiently than others by dedicated tools and algorithms. Furthermore, some applications, such as synthesis and probabilistic model checking, require that properties are represented as some type of deterministic $\omega$-automata. However, properties cannot always be represented by automata with the desired acceptance condition and determinism.

In this paper we study the problem of approximating linear-time properties by automata in a given class.  Our approximation is based on preserving the language up to a user-defined precision given in terms of the size of the finite lasso representation of infinite executions that are preserved. We study the state complexity of different types of approximating automata, and provide constructions for the approximation within different automata classes, for example, for approximating a given automaton by one  with a simpler acceptance condition. 

\end{abstract}

\section{Introduction}\label{sec:intro}
The specification of linear-time properties is a key ingredient of all typical frameworks for the verification and synthesis of reactive systems. The application of both automata-theoretic and symbolic algorithms requires that specifications are translated to some kind of $\omega$-automata. 
Depending on the considered problem, or on the applied methods and tools, there are often constraints on the type of the resulting  automaton, that is, on its acceptance condition, and on whether it is deterministic or not. For example, while for model checking of non-stochastic  systems it suffices to consider nondeterministic B\"uchi automata, synthesis and  probabilistic model checking require deterministic automata (e.g., deterministic parity automata). Furthermore, it is often the case that efficient specialized methods and tools exist for specific classes of automata i.e., specific  acceptance conditions. For instance, efficient synthesis algorithms exist for the class GR(1) of linear-time temporal logic specifications~\cite{BLOEM2012911}, which defines properties that are expressible as deterministic parity automata with three colors. 

Finding an equivalent automaton with a simpler acceptance condition is not always possible. The canonical example is the property defined by the linear-time temporal logic (LTL) formula $\LTLfinally\LTLglobally p$, for which no deterministic B\"uchi automaton exists. A more interesting example is given by the LTL formula $\varphi=(\LTLglobally\LTLfinally p \rightarrow \LTLglobally\LTLfinally q) \wedge (\LTLglobally\LTLfinally r \rightarrow \LTLglobally\LTLfinally s)$, which requires that if the proposition $p$ holds infinitely often then the proposition $q$ should hold infinitely often as well, and the same for the propositions $r$ and $s$. Requirements of this form occur often in the synthesis of reactive systems, but the formula $\varphi$ cannot be represented by a deterministic parity automaton with three colors, and cannot be transformed to a formula in the efficient class of GR(1) specifications. Moreover, automata with simpler acceptance conditions can often be larger in size than automata with more general acceptance conditions. For instance, there are languages for which deterministic Streett automata are exponentially smaller than  nondeterministic B\"uchi automata~\cite{phd-safra}.

Motivated by this,  we study the problem of approximating linear-time properties (respectively $\omega$-automata) by automata in a given class (respectively automata from a given subclass). The choice of language approximation is inspired by applications in bounded model checking~\cite{bmc} and bounded synthesis~\cite{boundedSynthesis}. These methods are based on the observation that for finite-state systems, it
suffices to consider lasso-shaped executions of bounded size. Our approximation exploits the same idea for the construction and transformation
of automata. 
Furthermore, equivalent $\omega$-regular languages share the same set ultimately-periodic words~\cite{Calbrix:1993:UPW:645738.666444}, and thus lasso-shaped words of bounded size provide an approximation to this set of words, one that improves when considering larger bounds on the size of lassos. 

Given an $\omega$-language  $L$ and a bound $n \in \nats$, we consider the  language $L_n$ of the ultimately-periodic words in $L$ representable in the form $u\cdot v^\omega$, and where $|u\cdot v| \leq n$. That is, the language $L_n \subseteq L$ consists of the words in $L$ representable as \emph{lassos of length n or smaller}. We are then interested in approximations of $L$ that are precise with respect to the language $L_n$, termed \emph{$n$-lasso-precise approximations}.

We study the properties of $n$-lasso-precise approximations across the three dimensions of the complexity of the automata for such languages: size, acceptance condition, and determinism. More precisely, we establish worst case bounds, in terms of $n$, on the size of automata for $n$-lasso-precise approximations. We also show that we can approximate a parity automaton with $m$ colors by one with $m' < m$ colors, with at most polynomial increase in the size of the automaton. For example, considering the formula $\varphi$ above, if we underapproximate the language of $\varphi$ with a language that is precise with respect to the set of words representable by lassos of length $n$ for a fixed $n$, we can represent the resulting language by a safety automaton (a parity automaton with one color). Furthermore, if, for example, $n=2$ the resulting automaton has $4$ states, while the minimal deterministic parity automaton for the language of $\varphi$ has $95$ states and $10$ colors.
We also study the approximation of nondeterministic by deterministic automata, and show that the  worst-case exponential blow-up in the size is unavoidable for $n$-lasso-precise approximations.

As another example, consider the property described by the LTL formula $(\LTLfinally\LTLglobally p )\wedge(\LTLglobally\LTLfinally q) $, where $p$ and $q$ are some atomic propositions. This is a conjunction of a stability property and a liveness property, which is also not expressible in the fragment GR(1). We can approximate this property by an $n$-lasso-precise deterministic B\"uchi automaton, enabling the application of efficient synthesis tools. Most importantly, unlike existing approaches, our method is not limited to approximating liveness properties by safety properties, which benefits the precision of the approximation.

\bigskip
The paper is structured as follows. In Section~\ref{sec:preliminaries} we start with a short background on linear-time properties and $\omega$-automata. In Section~\ref{sec:approx} we introduce the notion of $n$-lasso-precise approximation of linear-time temporal properties, and present all relevant automata constructions for these approximations. Here, we establish property-independent upper and lower bounds on the size of $\omega$-automata for $n$-lasso-precise approximations, and study the  overhead in terms of size incurred when approximating an automaton by one with a simpler acceptance condition. In Section~\ref{sec:experiments} we show that the problem of computing lasso-precise automata of bounded size for properties given as LTL formulas is in $\Sigma^P_2$. In Section~\ref{sec:conclusion} we conclude our results with a discussion on our approach and its potential for the development of new verification and synthesis algorithms.

\paragraph{\bf Related Work.}
Our definition of bounded lasso-precise approximation is motivated by bounded model checking~\cite{bmc}, bounded synthesis~\cite{boundedSynthesis}, synthesis for bounded environments~\cite{DBLP:conf/cav/DimitrovaFT19}, and synthesis of approximate implementations \cite{BoundedSystemsAndEnvironments}.  We extend these ideas of focusing on small counterexamples, small implementations, or bounded-state environments, respectively, to the realm of specifications.  

The structural complexity of $\omega$--automata has been studied in~\cite{LPAR-22:Why_These_Automata_Types,KrishnanPB95}, where the acceptance conditions of deterministic automata are related to their complexity. Here, on the other hand, we study complexity questions in the context of language approximations.

There is a rich body of work on the minimization of B\"uchi automata. Typical approaches, such as~\cite{DBLP:conf/concur/EtessamiH00,DBLP:conf/forte/GiannakopoulouL02,DBLP:conf/cav/SomenziB00} are based on merging states according to      simulation and bisimulation relations. In~\cite{OnTheVirtueOfPatience} the authors propose a  SAT-solver based minimization method. All these approaches consider language equivalence, while in this paper we study language approximation.

Reducing the size of automata by language approximation has been studied in the context of  languages over finite words. The approach in~\cite{FixingTheStateBudget} fixes a bound in the number of the states of a deterministic finite automaton for a safety language, and computes an automaton within that bound that approximates the original language. In addition to the fact that their method applies to languages over finite words, the key difference to our work is that while their goal is to optimize precision within a state budget, we approximate automata with ones with simpler acceptance conditions that guarantees a desired precision. In  descriptive complexity, there is a related notion to our $n$-lasso  precision, which is the notion of the automaticity~\cite{automaticity} of a language  which is the size of the minimal automaton that is precise for that language on words of length up to a given bound $n$. As automaticity is defined for finite-word languages, $n$-lasso precision can be seen as lifting these ideas to $\omega$-languages. 

The approximation of $\omega$-regular properties by ones with simpler acceptance conditions has not been, to the best of our knowledge, systematically studied so far. Standard approaches, such as 
\cite{DBLP:journals/sttt/SchuppanB04,DBLP:journals/entcs/BiereAS02}, approximate  liveness and other temporal properties via safety properties. In contrast,  our approximation allows us to approximate temporal properties with other temporal properties that are not necessarily safety.

\section{Preliminaries}\label{sec:preliminaries}
\paragraph*{Linear-time Properties and Lassos.}
 A \emph{linear-time property} $\varphi$  over an alphabet~$\Sigma$ is a set of infinite words $\varphi \subseteq \Sigma ^\omega$. 
Elements of $\varphi$ are called \emph{models} of $\varphi$.
A \emph{lasso of length $n$} over an alphabet $\Sigma$ is a pair $(u,v)$ of finite words $u\in \Sigma^{*}$ and $v \in \Sigma^{+}$ with  $|u\cdot v|~=n$ that induces the ultimately-periodic word  $u \cdot v^\omega$. We call $u\cdot v$ the \emph{base} of the lasso or ultimately-periodic word, and $n$ the \emph{length of the lasso}. The set $\mathit{Bases}(\varphi,n)$ is the set of bases of lassos of length $n$ that induce words that are models of $\varphi$. 

For a bound $n \in \nats$, we define the language $L_n(\varphi)=\{\sigma \in \Sigma^\omega\mid \exists u\cdot v \in \mathit{Bases}(\varphi,n).~ \sigma=u\cdot v^\omega \}$
as the language of models of $\varphi$ that can be represented by lassos of length $n$. 
We call the elements of $L_n(\varphi)$ the \emph{$n$-models} of $\varphi$.

If a finite word $w\in \Sigma^*$ is a prefix of a word $\sigma \in \Sigma^*\cup\Sigma^\omega$, we write $w\preceq \sigma$. For a language $L \subseteq \Sigma^* \cup \Sigma^\omega$, we define  $\mathit{Prefix}(L) = \{w \in \Sigma^* \mid \exists \sigma \in L: w \preceq \sigma\}$ as the set of all finite words that are prefixes of words in the language $L$.
For a word $w = \alpha_1\alpha_2\ldots\alpha_n\in \Sigma^*$ we define $w(i) = \alpha_i$ for each $i \in \{1,\ldots, n\}$.

\paragraph*{Automata Over Infinite Words.}
A \emph{nondeterministic parity automaton} over an alphabet $\Sigma$ is a tuple $\mathcal{A} =
(Q,Q_0,\delta,\mu)$, where
$Q$ denotes a finite set of states, $Q_0 \subseteq Q$ denotes a set of
initial states, $\delta: Q \times \Sigma \rightarrow
\mathcal{P}(Q)$ denotes a transition function that maps a state and an input letter to a
set of states, and $\mu: Q
\rightarrow C \subset \mathbb{N}$ is a coloring function with a finite set of colors $C$. 

A \emph{run} of $\mathcal{A} = (Q,Q_0,\delta,\mu)$ on an infinite word $\sigma = \alpha_1\alpha_2\dots \in \Sigma^\omega$ is an infinite sequence $\rho=q_0 q_1 q_2\ldots\in Q^\omega$ of states such that $q_0 \in Q_0$, and for every $i \in \nats$ it holds that $q_{i+1} \in \delta(q_i,\alpha_{i+1})$. A run $\rho=q_0q_1q_2\ldots$ is \emph{accepting} if it satisfies the \emph{parity condition}, which requires that the highest number occurring
infinitely often in the sequence $\mu(q_0)\mu(q_1)\mu(q_2)\dots\in C^\omega$ is even.  An infinite word $\sigma$ is \emph{accepted} by an automaton $\A$ if there exists an accepting run of $\A$ on $\sigma$.
The set of infinite words accepted by an automaton $\mathcal{A}$ is
called its \emph{language} $L(\mathcal{A})$. 

We say that a run $\rho$ \emph{has size $n \in \nats$} if $\rho$ is an ultimately-periodic run and $n$ is the smallest natural number such that $\rho = \rho_1\cdot(\rho_2)^\omega$ and  $|\rho_1\cdot\rho_2| = n$.

An automaton is  \emph{deterministic} if $|Q_0| = 1$, and for all states $q$ and input letters $\alpha$, $|\delta(q,\alpha)| \leq 1$. For a deterministic automaton we will see $\delta$ as a partial function $\delta: Q \times \Sigma \rightarrow Q$. We use the notation $\delta(q,\alpha)=\emptyset$ to denote that state $q$ has no successor for the letter $\alpha$.
We define the size $|\A|$ of an automaton $\A$ to be the number of its states, i.e., $|\A| = |Q|$.

A parity automaton is called a \emph{B\"uchi automaton} if and only if the image of
$\mu$ is contained in $\{1,2\}$, 
and a \emph{safety}
automaton if the image of $\mu$ is $\{0\}$.  B\"uchi 
automata are denoted by $(Q,Q_0,\delta,F)$, where
$F\subseteq Q$ denotes the states with the higher color. Safety
automata are denoted by $(Q,Q_0,\delta)$. A run 
of a B\"uchi automaton is thus accepting, if it contains infinitely many visits to $F$.
For safety automata, every infinite run is accepting.

We define an \emph{automaton type} to indicate whether the automaton is deterministic or nondeteministic, and its acceptance condition. We abbreviate deterministic as D  and nondeterministic as  N. For the acceptance conditions 
we use the abbreviations P (parity) and B (B\"uchi). Thus, for example, DPA stands for deterministic parity automaton, while NBA stands for Nondeterministic B\"uchi automaton.

\section{Lasso-precise Approximations of Linear-time Properties}\label{sec:approx}
We begin this section with a formal definition of the approximation of linear-time properties discussed in the introduction. More precisely, we introduce the notion of \emph{lasso-precise under-} and \emph{overapproximation} of a linear-time property~$\varphi$ for a given bound $n \in \nats$, in which we underapproximate (overapproximate) $\varphi$ with a linear-time property that has the same $n$-models as  $\varphi$. That is, the approximation is precise for $n$-models.

\subsection{Lasso-precise Approximations}

\begin{definition}[Lasso-precise Underapproximation]
For a bound $n \in \nats$, we say that a linear-time property $\varphi'$ is an $n$-lasso-precise underapproximation of a linear-time property $\varphi$, denoted $\varphi' \underapprox \varphi$, if $\varphi'\subseteq \varphi$ and $L_n(\varphi') = L_n(\varphi)$.
\end{definition}

\begin{definition}[Lasso-precise Overapproximation]
For a bound $n \in \nats$, we say that a linear-time property $\varphi'$ is an $n$-lasso-precise overapproximation of a linear-time property $\varphi$, denoted $\varphi' \overapprox \varphi$, if $\varphi'\supseteq \varphi$ and $L_n(\varphi') = L_n(\varphi)$. 
\end{definition}

In the rest of the paper we focus on underapproximations. All the results extend easily to lasso-precise overapproximations. In fact, if we have also the complement language of $\varphi$, an $n$-lasso-precise overapproximation of a  property $\varphi$ can be computed by computing an $n$-lasso-precise underapproximation of the complement of $\varphi$.

In the next sections we show how to construct automata for $n$-lasso-precise approximations of linear-time properties. For a property $\varphi$ the automata will recognize the language $L_n(\varphi)$.  This language includes also all words in $\varphi$ that are representable by a lasso of size $n'\leq n$, a fact that we establish with the next lemma.  
\begin{lemma}
	For any linear-time property $\varphi$ and bounds $n,n'\in \nats$, we have that $L_n(\varphi) \subseteq L_{n'}(\varphi)$, if  $n\leq n'$. 
\end{lemma}
\begin{proof}
	Every lasso of length $n$ can be unrolled to a lasso of length $n'$ by unrolling the loop $n'-n$ times. \qed
\end{proof}

\subsection{The Size of Lasso-precise Automata for Linear-time Properties} 

Since for any $\varphi$ the language 
$L_n(\varphi)$ is a safety language, we can always construct a deterministic safety automaton that is $n$-lasso-precise. In the following  we provide a construction which yields a deterministic safety automaton for a language $L_n(\varphi)$, and establish a lower bound on the size of an automaton for $L_n(\varphi)$.

\begin{theorem}[Safety automata for $n$-lasso-precise approximations]
	For every linear-time property $\varphi$ over an alphabet $\Sigma$ and a bound $n\in \nats$, there is a deterministic safety automaton $\mathcal A$ of size $O(|\Sigma|^n \cdot 2^{n\log n})$, such that $L(\A) \underapprox \varphi$.
        \label{thm:safety}
\end{theorem} 
\paragraph*{Idea \& Construction.} The automaton $\A$  accepts a word in two phases. 
The states used in the first phase are of the form $w \cdot \#^{m-1} \in (\Sigma \cup \{\#\})^n$, where $w$ is the portion of the prefix of length $n$ of the input word that has been read so far.
 In this phase, the automaton reads the prefix of length $n$ and stores it in the automaton state. Once the whole prefix is read, it checks whether the prefix of length $n$ is in $\mathit{Bases}(\varphi,n)$. If this is the case, then it transitions to the second phase, and checks if the word being read is an $n$-lasso, with this base.  

The states in the second phase are of the form 
$(w,(t_1,\ldots,t_n)) \in  \Sigma^n\times \{-,1,\dots,n\}^n$ , where $w \in \Sigma^n$ is the prefix read in the first phase, and $(t_1,\ldots,t_n)$ are indices of letters in $w$, whose role is explained below. To check that the word is an $n$-lasso, the automaton has to check if for some $\ell \in \{1,\ldots,n \}$ the input word is of the form $w(1)\ldots w(\ell-1)(w(\ell)\ldots w(n))^\omega$, that is, there is an $\ell$ which is a loop start position. To this end, the automaton tracks the possible loop start positions, starting with all positions, and for each new letter $\alpha$ it eliminates those positions that are not compatible with $\alpha$. 
More precisely, if the automaton reads a letter $\alpha$ in state $(w,(t_1,\ldots,t_n))$, it uses each $t_i$ to check whether the loop can start in position $i$ of $w$. 
Intuitively, $t_i$ is a position in $w$ that points to the letter that has to be read  next in order for $i$ to still be a possible loop start position. If the next letter $\alpha$ is not the same as $w(t_i)$, then $i$ cannot be a loop start position, and $t_i$ is eliminated by replacing it by $-$. Otherwise, $t_i$ is incremented, or set back to the loop start $i$ if the end of $w$ is reached.
A run of $\A$ is accepting if it never reaches a state $(w,(-,\ldots,-))$, that is,  a state in which each position is no longer a possible start of a loop.

Formally, the states of the automaton are given by $\mathcal A= (Q,\{q_0\},\delta)$ where: 
\begin{itemize}
	\item $Q= Q_1 \cup Q_2$, where $Q_1 = (\Sigma \cup \{\#\})^n$ and 
	$Q_2 =  \Sigma^n\times \{-,1,\dots,n\}^n$	

	\item In the initial state no letter has been read: $q_0= \#^n$.
	\item The transition relation $\delta$ is defined as follows.
	\begin{itemize}
		\item In the first phase if we are at a state $q=w\cdot \#^m$ for some  $1<m\leq n$ and $w \in \Sigma^{n-m}$, then $$\delta(q,\alpha)= w\cdot\alpha\cdot \#^{m-1} $$
		\item In the transition between the first and the second phase, which happens once the prefix of length $n$ has been read, and when we  are at a state   $q = w\cdot \#$ for some  $w \in \Sigma^{n-1}$ the transition is given by 
		$$\delta(q,\alpha) = (w\cdot\alpha,(t_1,\dots,t_n))$$
		 where 
		$$t_i = \begin{cases}
 												i &  w(1)\dots (w(i)\dots w(n))^\omega  \in \varphi \\
 												- & \text{ otherwise}\\
 											\end{cases}$$
 											Note that determining the successor state in this case requires checking if a given word is in $\varphi$. 
 			Initially, only loop start positions  $i$ for which
$w(1)\dots (w(i)\dots w(n))^\omega \in \varphi$
are allowed, so the second phase starts with state  $(w,(t_1,\ldots,t_n))$, in which each pointer $t_i$ points to the start of the corresponding loop if $w(1)\dots (w(i)\dots w(n))^\omega \in \varphi$, and is set to $-$ otherwise.

 		\item In the second phase, for a state $q= (w,(t_1,\dots,t_n))$ with $w \in \Sigma^n$ and where there exists $i\leq n$ with $t_i\not = - $, the transition for such a state is given by 
$$\delta(q,\alpha) = (w,(t'_1,\dots,t'_n)) $$ 
where 
$$t'_i = \begin{cases}
												- &  t_i = - \\
												  &  \text{or }w(t_i) \not = \alpha\\
												  &\\
												t_i+1 & t_i < n ~\wedge~ w(t_i)=\alpha\\
												&\\
												i & t_i = n ~\wedge~ w(t_i)=\alpha
											 \end{cases}	$$

Here we track valid loop start position as follows. If $\alpha\neq w(t_i)$, then the loop start $i$ is eliminated by replacing $t_i$ by $-$. Otherwise, we move the pointer one step to the right by incrementing $t_i$. In case $t_i$ is equal to $n$, i.e., at the end of the lasso, $t_i$ is reset to the corresponding loop start position $i$.
\item If only $-$ remain in the tuple $(t_1,\dots,t_n)$, the automaton rejects
	$$\delta((w,(-,\dots,-)),\alpha) = \emptyset$$
	for any $w \in \Sigma^n$.
	\end{itemize}
\end{itemize}

The number of states in $Q_1$ is $(|\Sigma|+1)^n$, and for $Q_2$ it is $|\Sigma|^n \cdot (n+1)^n$. 
	\qed
\vskip 0.2cm

The number of states of the deterministic safety automaton defined above is exponential in the parameter $n$ on the length of the lassos for which the approximation should be precise. In the next theorem we exhibit a family of linear-time properties for which this exponent is unavoidable, that is, the minimal $n$-lasso-precise NPA has size exponential in $n$.

\begin{theorem}
	There is a family of linear-time properties $\varphi_n$ for $n \in \nats$ over an alphabet $\Sigma$, such that, every parity automaton that is $n$-lasso-precise for $\varphi_n$ has at least $|\Sigma|^n$ states. \looseness=-1 
\end{theorem}
\begin{proof}
Let $\Sigma$ be an alphabet. We define  $\varphi_n = \{\sigma^\omega \mid  \sigma\in \Sigma^n\}$ for  $n \in \nats$. We  show that the family $\varphi_n$ of linear-time properties has the required properties.

Fix $n \in \nats$, and consider the language $ \varphi_n $. By definition of $\varphi_n$, every lasso-precise automaton for $\varphi_n$ for the bound $n$ is in fact an automaton for $\varphi_n$. Let $\mathcal A = (Q,Q_0,\delta,\mu)$ be a nondeterministic parity automaton for $\varphi_n$.  
For each  $\sigma^\omega \in \varphi_n$ there exists  at least one accepting run $\rho=q_0 q_1 q_2,\ldots$ of $\mathcal A$ on $\sigma^\omega$. We denote with $q(\rho,n)$ the state $q_n$ that appears at the position indexed $n$ of a run $\rho$.  Let us define the set
\[Q_n  =\{q(\rho,n) \mid  \exists \sigma^\omega \in \varphi_n:\; \rho \text{ is an accepting run of }\mathcal A \text{ on } \sigma^\omega \}.\]
That is, $Q_n $ consists of the states that  appear at position $n$ on some accepting run on some word from $\varphi_n$. We will show that  $|Q_n| \geq |\Sigma|^n$. 

Assume that this does not hold, that is,  $|Q_n| < |\Sigma|^n$. Since $|\varphi_n| = |\Sigma|^n $, this implies that there exist  $\sigma_1,\sigma_2 \in \Sigma^n$, such that $\sigma_1 \neq \sigma_2$ and there exists accepting runs $\rho_1$ and $\rho_2$ of $\mathcal A$ on $\sigma_1^\omega$ and $\sigma_2^\omega$ respectively, such that $q(\rho_1,n)  = q(\rho_2,n)$. 
That is, since we assumed that the number of states in $Q_n$ is smaller than the number of words in $\varphi_n$, there must be two different words who have accepting runs visiting the same state at position $n$. We now construct a run $\rho_{1,2}$ that follows $\rho_1$ for the first $n$ steps, ending in state $q(\rho_1,n)$, and from there on follows $\rho_2$. It is easy to see that $\rho_{1,2}$ is a run on the word $\sigma_1 \cdot \sigma_2^\omega$. It is accepting, since $\rho_2$ is accepting. This is a contradiction, since $\sigma_1 \cdot \sigma_2^\omega \not\in  L(\mathcal A)$ as $\sigma_1 \neq \sigma_2$.

Thus, we have shown that $|Q| \geq  |Q_n| \geq |\Sigma|^n$. Since $\mathcal A$ was an arbitrary NPA for $\varphi_n$, this implies that the minimal NPA for $\varphi_n$ has at least $|\Sigma|^n$ states.\qed
\end{proof}

In the theorems above we established an upper and a lower bound on the size of automata for $n$-lasso-precise approximations. These bounds are independent of the way the original language is represented. If a language $L$ is given as an $\omega$-automaton, this automaton is clearly an automaton for the most precise $n$-lasso-precise underapproximation of $L$. In practice, however, we might be interested in finding a smaller/minimal automaton of the same type for an $n$-lasso-precise approximation of $L$. Note that the minimal $n$-lasso-precise automaton of the same type will never be larger than the given automaton.

\subsection{Lasso-precise Approximations with Simpler Acceptance Conditions}
We now turn to establishing the upper bounds for approximating B\"uchi automata with safety automata, and, more generally, approximating parity automata with parity automata with fewer colors.
More precisely, we present constructions for approximating linear-time properties with automata with certain acceptance conditions and show that the size of the constructed automaton is polynomial in the size of an automaton for the original property.
\begin{theorem}[Approximating B\"uchi automata by safety automata]\label{thm:buechi-to-safety}
	For every (deterministic or nondeterministic) B\"uchi automaton $\A=(Q,Q_0,\delta,F)$  and a bound $n\in \nats$, there is a (deterministic or nondeterministic, respectively) safety automaton $\A'$ with $n \cdot |Q\setminus F|^2+|F|$ states, such that, $L(\A') \underapprox L(\A)$.
\end{theorem}
\paragraph*{Idea \& Construction.}
	 We construct a safety automaton $\A'$ using the following idea: If an ultimately-periodic word $\sigma = u\cdot v^\omega$ with $|u\cdot v| = n$ is accepted by a B\"uchi automaton $\A =  (Q,Q_0,\delta,F)$, then $\A$ has a run for $\sigma$, where it takes no more than $n\cdot |Q\setminus F|$ steps to observe a state in $F$, and, furthermore $F$ is visited at least once every $n\cdot |Q\setminus F|$  steps. In the automaton $\A'$, we keep track of the number of steps without seeing an accepting state, and reset the counter every time we visit one. If the counter exceeds  $n\cdot |Q\setminus F|$, then $\A'$ rejects. 
	
	Formally, we define $\mathcal A'= (Q',Q'_0,\delta')$ as follows:
	\begin{itemize}
		\item $Q'= ((Q\setminus F)\times\{1,\dots,n\cdot|Q\setminus F|\}) \cup (F\times\{0\})$
		\item $Q'_0 = (Q_0\cap F)\times\{0\} \cup (Q_0\setminus F)\times\{1\}$ 
		\item For the transition relation we distinguish two cases. For $c< n\cdot|Q\setminus F|$  
		$$\delta((q,c),\alpha) = \{(q',d) \mid q'\in \delta(q,\alpha), ~d= 0 \mbox{ if } q' \in F,~ d= c+1 \mbox{ if } q'\not \in F\}$$ 
		otherwise $\delta((q,c),\alpha) = \emptyset$.
		\end{itemize}  

	Note that, if the given B\"uchi automaton is deterministic, then our construction also produces a deterministic safety automaton.   \qed
\bigskip

Theorems~\ref{thm:safety} and \ref{thm:buechi-to-safety} provide safety automata of different sizes: the safety automaton obtained by Theorem~\ref{thm:safety} is exponential in the bound, the safety automaton obtained by
Theorem~\ref{thm:buechi-to-safety} is linear in the bound. The reason for this difference is that the size of the automaton constructed according to Theorem~\ref{thm:safety}  is independent of the linear-time property,
whereas the size of the automaton constructed according to Theorem~\ref{thm:buechi-to-safety} is for a specific linear-time property
(given as a B\"uchi automaton, whose size enters as a quadratic factor). The following theorem shows that a further reduction, below the linear number of states in the bound, is impossible.

\begin{theorem}
\label{thm:buechi-to-safetylower}
There is a linear-time property $\varphi$, such that, for every bound $n \in \mathbb N$,  
 every safety  $n$-lasso-precise automaton for $\varphi$  has at least $n$ states. 
\end{theorem}
\begin{proof}
Let $\Sigma=\{0,1\}$. We define $\varphi$ as the language over $\Sigma$ that consists of all words where the letter $1$ occurs infinitely often.
 Let $\A = (Q,Q_0,\delta)$ be a safety $n$-lasso-precise automaton for $\varphi$.
We consider the set $Q' \subseteq Q$ of states on the first $n$ positions of an accepting run of the word $(0^{n-1}1)^\omega$.
We show that $|Q'|=n$ and, therefore, $|Q|\geq n$.

Assume that this does not hold, i.e., $|Q'|<n$; then some state $q$ must appear
on two different positions among the first $n$ positions of the run. By repeating the part of the run between the two occurrences of $q$ infinitely often,
we obtain an accepting run for the word $0^\omega$, which contradicts our assumption that $\A$ is $n$-lasso-precise for $\varphi$.\qed
\end{proof}

With a construction similar to Theorem~\ref{thm:buechi-to-safety}, we can approximate a parity automaton with $m+1$ colors by a parity automaton with $m$ colors.\looseness=-1

\begin{theorem}[Approximating parity automata by parity automata with one color less]
\label{thm:onecolorlessParity}
	For every deterministic  parity automaton $\A =(Q,Q_0,\delta,\mu)$  with $m+1$ colors and a bound $n\in \nats$, there is a deterministic parity automaton $\A'$ with $m$ colors and $n \cdot |Q\setminus F|^2+|F|$ states, where $F$ is the set of states with highest color, such that $L(\A') \underapprox L(\A)$.
\end{theorem}

By iteratively applying Theorem~\ref{thm:onecolorlessParity},  we can approximate any parity automaton with $m$ colors by a corresponding parity automaton with $m' < m$ colors. This, however, will incur a blow-up in the size of the automaton that is exponential in the number $m$ of colors. We now provide a direct construction, which is polynomial both in $m$ and in the size of $\A$.

\begin{theorem}[Approximating parity automata by parity automata with fewer colors]
	For every deterministic  parity automaton $\mathcal A=(Q,Q_0,\delta,\mu)$  with $m$ colors, a bound $n\in \nats$ and $0<m'<m$, there is a deterministic  parity automaton $\mathcal A'$  with $m'$ colors  and $(n\cdot|Q|+1)\cdot |Q|\cdot (m-m'+2)$ states such that $L(\A') \underapprox \mathcal L(\A)$.
        \label{thm:fewercolors}
\end{theorem}
\paragraph*{Idea \& Construction.}
Our automaton construction is based on the following idea. An ultimately-periodic word in $L(\A)$ representable by a lasso  of length $n$ has an ultimately-periodic run in $\A$ of size at most $n\cdot |Q|$. The ultimately-periodic run is accepting if the highest color occurring in its period is even. For a given ultimately-periodic word with lasso of length $n$, our constructed automaton $\A'$ checks whether this word has an ultimately-periodic accepting run of size $n\cdot |Q|$ in $\A$.
Adapting the same idea as in  Theorem~\ref{thm:buechi-to-safety}, we check whether the colors we wish to eliminate appear within $n\cdot |Q|$ steps. We reject words with runs where these colors appear with distances larger than $n\cdot |Q|$. On the other runs we use the acceptance condition of the remaining colors. 

Let $\A = (Q,Q_0,\delta,\mu)$ where $\mu: Q \rightarrow \{0,\dots, m-1\}$. We construct the parity automaton $\A' =(Q',Q'_0,\delta',\mu')$ with $\mu': Q \rightarrow \{0,\dots,m'-1\}$ and where: 
$$Q' = (Q\times \{0,\dots,n\cdot|Q|\}) \cup  
 	(Q\times \{0,\dots,n\cdot|Q|\}\times \{-1,m',\dots,m-1\})$$
 and 
$$Q'_0 = \{(q,0) \mid  q \in Q_0\}.$$ 
The transition relation and coloring function are given as follows. 
In contrast to Theorem~\ref{thm:buechi-to-safety} we  now need to first check which is the highest color that appears in the period of the run.  
This check is done respecting the following cases. 
 $$\mbox{Case (1):  }\delta'((q,c),\alpha) = \{(q',c+1) \mid q' \in \delta(q,\alpha) \} ~~\mbox{ if }\; c< n\cdot|Q|-1
									$$
As we are only interested in the highest color that appears in the period of the run, case (1)  makes sure that we reach this period by skipping the first $n\cdot |Q|$ steps, i.e., we simply follow the transition relation of $\A$ and increase the counter (denoted by $c$).
\begin{align*}
	\mbox{Case (2):  } \delta'((q,c),\alpha)= & \{(q',0,\mu(q')) \mid q' \in \delta(q,\alpha),\mu(q')\ge m'\} ~\cup \\
	& \{(q',0,-1)  \mid q' \in \delta(q,\alpha),~\mu(q')< m'\}~~\mbox{ if }  c= n\cdot|Q|-1
\end{align*}
In Case (2) is the transition to the second phase, once we have skipped the first $n\cdot |Q|$ states. From here on we save the highest color seen that is larger than $m'-1$.
\begin{align*}
	\mbox{Case (3):  } \delta'((q,c,h),\alpha)= & \{(q',0,\mu(q')) \mid q' \in \delta(q,\alpha),\mu(q')> h, \mu(q')\ge  m'\} ~\cup \\
	& \{(q',0,h)  \mid q' \in \delta(q,\alpha),~\mu(q') =h\}~\cup \\
	& \{(q',c+1,h) \mid q' \in \delta(q,\alpha),~ \mu(q')<h \vee \mu(q')<m'\}\\
	& \mbox{ if }  c\leq n\cdot|Q|-1
\end{align*}
In case (3) we track the highest color $h$ seen so far. If $h$ is higher than $m'-1$ we save this color  and check how long it takes for this color to reappear. In case it appears in less that $n\cdot|Q|$ steps ($\mu(q')=h$) we reset the counter for this color. If a higher color is observed ($\mu(q')>h$), $h$ is replaced by the color  and the counter is reset. 
\begin{align*}
\mbox{Case (4):  }\delta'((q,c,h),\alpha)	 = & \{(q',c,h) \mid q' \in \delta(q,\alpha),~ \mu(q')<m'\}\\
& \mbox{if  } c = n\cdot |Q| \mbox{ and } h = -1 
\end{align*}
In the case where the counter exceeds $n\cdot |Q|$ for some saved color, the automaton rejects, but only if colors higher than $m'$ were observed along the way. Otherwise, the automaton $\A'$ accepts as $\A$ with the non-eliminated colors.  The coloring function is defined as follows 
$$\mu'(\tilde q) = \begin{cases}
							0 & \tilde q = (q,c)\\
							1 & \tilde q = (q,c,h),\; 0 \leq c < n\cdot |Q|, h \text{ is odd}\\
							0 & \tilde q = (q,c,h),\; 0 \leq c<n\cdot |Q|, h \text{ is even}\\
							\mu(q) & \tilde q = (q,c,-1),\; c = n\cdot |Q|
						 \end{cases}$$

\hfill \qed

With this, we conclude the study of the approximation of linear-time properties represented by  $\omega$-automata with  lasso-precise automata with simpler acceptance conditions preserving their determinism. In the next subsection, we turn to the approximation of nondeterministic automata with lasso-precise deterministic automata.

\subsection{Lasso-precise Deterministic Approximations}
We now study lasso-precise approximations from the point of view of the determinism of the automata representing  $\omega$-regular languages.
The complexity of determinizing $\omega$-automata, in particular the construction of deterministic parity automata, has been studied extensively (cf.~\cite{DeterminisationParityScheweV14}). The size of the deterministic automaton that recognizes the
same language as the given nondeterministic automaton is, in the worst case, exponential in the number of states of the given automaton.
By contrast, the size of the deterministic safety automaton provided by Theorem~\ref{thm:safety} is independent of the given language and exponential only in the bound. 
For small bounds, Theorem~\ref{thm:safety} thus provides a deterministic lasso-precise approximation with a small number of states. The following theorem shows that, for large bounds, it is not, in general, possible to produce small deterministic lasso-precise approximations.
If the bound is as large as the number of states of the given nondeterministic automaton, then the deterministic lasso-precise approximation has, in the worst case, an exponential number of states.

\begin{theorem}
For every $k \in \nats$ there exists a nondeterministic parity automaton $\A$  with $O(k)$ states, such that, for every bound $n  \geq |\A|$, the minimal deterministic parity automaton $\A'$ with  $L(\A') \underapprox L(\A)$ has at least $2^{k}$ states. 
\end{theorem}
\begin{proof}
Let $\Sigma = \{0,1,2\}$, and consider the language
 \[\Omega = \{\{0,1\}^i \cdot 1 \cdot \{0,1\}^{(k-1)} \cdot 2 \cdot 1^\omega \mid i < k\}.\]
That is, $\Omega$ consists of the infinite words over $\{0,1,2\}$ in which the letter $2$ appears exactly once, and the letter exactly $k$ positions prior to that is a $1$, preceded by at most $k-1$ letters.

We can construct a  nondeterministic parity automaton $\mathcal A = (Q,Q_0,\delta,\mu)$ for $\Omega$ with $2k+1$ states as follows. We let $Q = \{0,1\} \times \{1,\ldots,k\} \cup \{q_a\}$ and $Q_0 = \{(0,1)\}$. The function $\mu$ is such that $\mu(q_a)=0$, and $\mu(q) = 1$ for all $q \neq q_a$. We define the transition relation $\delta$ such that $\delta(q_a,1) = \{q_a\}$  and $\delta(q_a,\alpha)= \emptyset$ if  $\alpha \in \{0,2\}$ and 
 \[\delta((b,i),\alpha) = \begin{cases}
 \{(0,i+1)\} & \text{if } b=0, i < k, \alpha=0,\\
 \{(0,i+1),(1,1)\} & \text{if } b=0, i < k, \alpha=1,\\
 \{(1,1)\} & \text{if } b=0, i = k, \alpha=1,\\
 \{(1,i+1)\} & \text{if } b=1, i < k,\alpha \not = 2\\
 \{q_a\} & \text{if } b=1, i = k, \alpha=2.
 \end{cases}\]
 
Let $n \in \nats$ be a bound such that  $n \geq 2k+1$, and let $\mathcal A'$ be a DPA such that $L(\A') \underapprox L(\A)$. By the definition of  $\Omega$ and the fact that $n  \geq 2k+1$ we have that $L(\A') = \Omega$.  We will show that $\A'$ has  at least $2^k$ states.

Suppose that $|\A'| < 2^k$. This means that there exist two different words $\sigma_1,\sigma_2\in \{0,1\}^k$ such that  $\A'$ ends up in the same state when run on $\sigma_1=\alpha_{1,1}\ldots\alpha_{1,k}$ and when run on $\sigma_2=\alpha_{2,1}\ldots\alpha_{2,k}$. 
Since $\sigma_1$ and $\sigma_2$ are different, there must exist an $i$ such that  $\alpha_{1,i} \neq \alpha_{2,i}$. W.l.o.g., suppose that $\alpha_{1,i} =1$  and $\alpha_{2,i}=0$.
 Let $\sigma = 1^{i-1}\cdot 2\cdot 1^\omega$. Consider the words $\sigma_1\cdot\sigma$ and $\sigma_2\cdot\sigma$. In $\sigma_1\cdot\sigma$, the letter appearing $k$ positions before the letter $2$ is $1$, and in $\sigma_2\cdot\sigma$ this letter is $0$. 
 Thus, by the definition of $\Omega$ and $\A'$ we have that $\sigma_1\cdot\sigma$ must be accepted by $\A'$, and $\sigma_2\cdot\sigma$ must be rejected, which is a contradiction with the fact that $\A'$ is deterministic and the assumption that $\sigma_1$ and $\sigma_2$ lead to the same state. 

Since $\A'$ is an arbitrary deterministic parity automaton such that $L(\A') \underapprox L(\A$), we conclude that the minimal such automaton has at least $2^k$ states. \qed
\end{proof}

\section{Automata with Bounded Size}\label{sec:experiments}

In many cases, one is interested in constructing an  automaton of minimal size for a given language. In this section, we solve the problem of computing $n$-lasso-precise automata of \emph{bounded size}. By iteratively increasing the bound on the size of the automaton, this approach can be used to
construct minimal automata. 

Here we consider languages given as LTL formulas \cite{LTL}. LTL formulas are a common starting
point for many verification and synthesis approaches. Rather than going through an intermediate precise automaton, here we propose a symbolic approach that directly yields  an automaton whose language is $n$-lasso-precise approximation for the LTL formula.

\begin{theorem}
	For a  linear-time property $\varphi$  given as an LTL formula over $\AP$, and given bounds $n$, $k$ and $m$, deciding whether there  exists a deterministic parity automaton $\A$ of size $k$ and number of colors $m$, such that, $L(\A) \underapprox \varphi$ is in $\Sigma_{2}^P$.
\end{theorem}
\begin{proof}
	We show that the problem can be encoded by a quantified Boolean formula  with one quantifier alternation (2-QBF) of size polynomial in the length of the LTL formula $\varphi$, and the bounds $k,n$ and $m$. Deciding quantified Boolean formulas in the 2-QBF fragment is in $\Sigma^P_2$ \cite{handbookSatQBF}.
\paragraph{Construction.}
		\begin{flalign*}
	&\exists \{\delta_{s,\alpha,s'} \mid s,s' \in Q, \alpha \in 2^\AP\}.\\
	& \exists \{\mu_{s,c} \mid s \in Q, 0\leq c<m\}& \\
	&\forall\{a_j \mid a \in \AP, 0\leq j < \max\{k,n\}\}.\\
	&\forall\{l_j \mid 0\leq j < \max\{k,n\}\}& \\
	&\forall \{s_j \mid s\in Q, 0\leq j < n\cdot k\}.\\
	&\forall \{r_j \mid 0\leq j < n\cdot k\} &\\
	& \phi^{k,m}_\text{DPA} \wedge (\phi_\text{loop} \rightarrow \phi_{\A \subseteq \varphi} \wedge \phi_{=_n})& 
\end{flalign*}
where \begin{itemize}
	\item $\phi_{\A \subseteq \varphi} \equiv \phi_{k\text{-accrun}} \wedge \phi_{\text{match(k)}}\rightarrow \phi_{\in L_k(\varphi)}$
	\item $\phi_{=_n} \equiv \phi_{\in L_n(\varphi)} \wedge \phi_{\text{match}(n\cdot k)} \rightarrow \phi_{\in L_n(\A)}$
\end{itemize}

	 The formula encodes the existence of a deterministic parity automaton $\A=(Q,q_0,\delta,\mu)$ with $L(\A) \underapprox \varphi$. The transition relation of the automaton is encoded in the variables $\delta_{s,\alpha,s'}$ that define whether the automaton has a transition from state $s\in Q$ to state $s'\in Q$ with a letter $\alpha \in 2^\AP$. Additional variables $\mu_{s,c}$ define the coloring of the states of the guessed automaton. A variable $\mu_{s,c}$ is true if a state $s$ has color $c$. Using the constraint $\phi_{\text{DPA}}^{k,m}$ we force a deterministic transition relation for the automaton and make sure that each state has exactly one color
	
	The relation $\underapprox$ is encoded in the formula $\phi_{\text{loop}} \rightarrow \phi_{\A \subseteq \varphi} \wedge \phi_{=_n}$. To check whether $\A \underapprox \varphi$ we need to check that: (1) $\A$ is a strengthening of $\varphi$, i.e., $\A \subseteq \varphi$, and (2) $\A$ is precise up to ultimately-periodic words of size $n$, i.e., $L_n(A)= L_n(\varphi)$. The strengthening is encoded in the constraint $\phi_{\A \subseteq \varphi} \equiv \phi_{k\text{-accrun}} \wedge \phi_{\text{match(k)}}\rightarrow \phi_{\in L_k(\varphi)}$. To check whether $\A$ is a strengthening of $\varphi$ we need to check that all accepting ultimately-periodic runs of size $k$ of $\A$ induce ultimately-periodic words of size $k$ that satisfy $\varphi$. This is encoded in the formulas $\phi_{k\text{-accrun}}$, $\phi_{\text{match(k)}}$ and $ \phi_{\in L_k(\varphi)}$. 
	The formula $\phi_{\text{match(k)}}$ encodes an ultimately-periodic run in $\A$ of size $k$ using the variables $s_j$ for $0\leq j <k$ which determine which state of the automaton is at each position in the run   and variables $r_j$ which determine the loop of the run. The formula $\phi_{k\text{-accrun}}$ checks whether this run is accepting by checking the highest color in the period of the run.  If both formulas are satisfied then it remains to check whether the induced run satisfies $\varphi$, which is done using the constraint $\phi_{\in L_k(\varphi)}$.
	 The constraint resembles the encoding given in \cite{bmc} for solving the bounded model checking problem. It is defined over the variables $a_j$, where $a\in \AP$ and $0\leq j<k$ and the variables $l_j$ for $0\leq j <k$. A variable $a_j$ is true if the transition at position $j$ in the run that satisfies $\phi_{k\text{-accrun}}$ and  $\phi_{\text{match(k)}}$ represents a letter where $a$ is true and if $\varphi$ allows $a$ to be true at that position. Variables $l_j$ define the position of the loop of the ultimately-periodic word induced by the run.
	
	If $\A $ satisfies the strengthening condition it remains to check whether $\A$ accepts all ultimately-periodic words of size $n$ that satisfy $\varphi$. This condition is encoded in the constraint $\phi_{=_n} \equiv \phi_{\in L_n(\varphi)} \wedge \phi_{\text{match}(n\cdot k)} \rightarrow \phi_{\in L_n(\A)}$. If an ultimately-periodic word  of size $n$ encoded by the variables $a_j$ for $0\leq j< n$ and loop variables $l_j$ satisfies $\varphi$ (checked by the formula $\phi_{\in L_n(\varphi)}$), then we match this ultimately-periodic word to its run in $\A$ (using the formula $\phi_{\text{match}(n\cdot k)}$). Notice that we have to match the word to a run in $\A$ of size $n \cdot k$ as words of length $n$ might induce runs of size $n\cdot k$. If the latter formulas are satisfied it remains to check whether the run in the automaton is accepting. 
	 
	Finally, the formula $\phi_\text{loop}$ asserts that only one loop is allowed at a time. All formulas are of size polynomial in $k$, $n$, $m$ and $\varphi$. \qed
\end{proof}

The construction above can also be used for computing nondeterministic automata by changing the constraints on the transition relation of the automaton.

\section{Discussion}
\label{sec:conclusion}

The key idea behind algorithmic methods like bounded model
checking~\cite{bmc} and bounded synthesis~\cite{boundedSynthesis} is that, for finite-state systems, it
suffices to consider lasso-shaped executions of bounded size. The
notion of $n$-lasso-precise approximation, introduced in this paper,
exploits the same observation for the construction and transformation
of automata. 

The new constructions for $n$-lasso-precise underapproximations have attractive properties.
Theorem~\ref{thm:safety} shows that it is possible to approximate a
given language with a deterministic safety automaton whose size is
exponential in the bound, but independent of the given language. For
small bounds, any language can thus be effectively approximated by a
deterministic safety automaton.  Theorem~\ref{thm:fewercolors} shows that reducing
the number of colors of a parity automaton incurs at most a polynomial
increase in the number of states of the original automaton.


The results indicate significant potential for new verification and
synthesis algorithms that work with $n$-lasso-precise approximations
instead of precise automata. A key novelty is that our constructions
allow us to approximate a given temporal property with a property of a simpler type without necessarily reducing all the way to safety. For example, we can approximate a given temporal
property with a parity automaton with three colors, for which efficient synthesis algorithms exist~\cite{BLOEM2012911}.

The constructions of the paper allow us to directly construct automata
for the approximations. An interesting topic for future work is to
complement these constructions with fast techniques that reduce the
number of states of an automaton without necessarily producing a
minimal automaton. Similar techniques, which, however, guarantee full
language equivalence rather than $n$-lasso precision, are commonly
used in the translation of LTL formulas to B\"uchi automata
(cf.~\cite{LTLtoBuchi}).

\bibliographystyle{plain}
\bibliography{bib.bib}

\end{document}